\documentclass[runningheads]{llncs}

\usepackage{graphicx}
\usepackage{pifont}
\usepackage{verbatim}
\usepackage{amsmath,amsfonts}
\usepackage{enumerate,url}
\usepackage{paralist}
\usepackage{subfig}
\usepackage{wrapfig}

\usepackage{algorithm}
\usepackage[noend]{algpseudocode}

\begin{document}

\let\doendproof\endproof
\renewcommand\endproof{~\hfill\qed\doendproof}

\newcommand{\leftp}[1]{\texttt{left}(#1)}
\newcommand{\rightp}[1]{\texttt{right}(#1)}

\newcommand{\lefte}[1]{\texttt{leftEdge}(#1)}
\newcommand{\righte}[1]{\texttt{rightEdge}(#1)}

\newcommand{\vis}[1]{\texttt{Vis}(#1)}

\spnewtheorem{myremark}{Remark}{\bfseries}{\itshape}
\newcommand{\down}[1]{\raisebox{-.4ex}{#1}}

\makeatletter
    \renewcommand*{\@fnsymbol}[1]{\ensuremath{\ifcase#1\or *\or \dagger\or \S\or
       \mathsection\or \mathparagraph\or \|\or **\or \dagger\dagger
       \or \ddagger\ddagger \else\@ctrerr\fi}}
\makeatother

\newcommand*\samethanks[1][\value{footnote}]{\footnotemark[#1]}

\title{Guarding the Vertices of an Orthogonal Terrain using Vertex Guards}

\titlerunning{Guarding the Vertices of an Orthogonal Terrain using Vertex Guards}

\author{
Saeed Mehrabi\inst{1}
}

\authorrunning{S. Mehrabi}

\institute{
Cheriton School of Computer Science, University of Waterloo, Waterloo, Canada.
\email{smehrabi@uwaterloo.ca}
}


\newcommand{\keywords}[1]{\par\addvspace\baselineskip
\noindent\keywordname\enspace\ignorespaces#1}


\maketitle

\begin{abstract}
A \emph{terrain} $T$ is an $x$-monotone polygonal chain in the plane; $T$ is orthogonal if each edge of $T$ is either horizontal or vertical. In this paper, we give an exact algorithm for the problem of guarding the convex vertices of an orthogonal terrain with the minimum number of reflex vertices.
\end{abstract}

\section{Introduction}
\label{sec:introduction}
A \emph{terrain} $T$ is an $x$-monotone polygonal chain in the plane; $T$ is called \emph{orthogonal} if each edge of $T$ is either horizontal or vertical. A point guard $g$ can see a point $p$ on $T$ if the line segment $gp$ lies on or above $T$. Generally, there are two types of guarding problems on a (not necessarily orthogonal) terrain $T$; namely, the \emph{continuous} terrain guarding and the \emph{discrete} terrain guarding problems. In the continuous terrain guarding problem, the objective is to guard the entire $T$ with the minimum number of point guards while in the \emph{discrete} terrain guarding problem, we are given two sets $P$ and $G$ of points on $T$, and the objective is to guard the points in $P$ with the minimum number of guards selected from $G$.

\noindent\paragraph{\bf Related Work.} The standard art gallery problem, where the objective is to guard a polygon with the minimum number of point guards, was first introduced by Klee in 1973~\cite{orourke1987}. Chv\'{a}tal~\cite{chvatal1975} showed that $\lfloor n/3 \rfloor$ point guards are always sufficient and sometimes necessary to guard a simple polygon with $n$ vertices. For the orthogonal art gallery problem, Kahn et al.~\cite{jeff1983} proved that $\lfloor n/4 \rfloor$ guards are always sufficient and sometimes necessary to guard a simple orthogonal polygon with $n$ vertices. Lee and Lin~\cite{lee1986} showed that the art gallery problem is \textsc{NP}-hard on simple polygons. The problem is also \textsc{NP}-hard on orthogonal polygons~\cite{dietmar1995} and on monotone polygons~\cite{krohn2013}. If the polygon is simple, then the art gallery problem is \textsc{APX}-hard~\cite{eidenbenz2001}. If the polygon is allowed to have holes, then the problem cannot be approximated by a polynomial-time algorithm with factor $((1-\epsilon)/12)\ln n$ for any $\epsilon>0$, where $n$ is the number of the vertices of the polygon~\cite{eidenbenz2001}. Kirkpatrick~\cite{Kirkpatrick15} gave an $O(\log \log OPT)$-approximation algorithm for the problem using \emph{vertex} guards (i.e., the point guards are restricted to be placed on the vertices of the polygon) on arbitrary simple polygons, where $OPT$ denotes the size of an optimal solution. Krohn and Nilsson~\cite{krohn2013} gave a constant-factor approximation algorithm for the problem using vertex guards on monotone polygons. They also gave an $O(OPT^2)$-approximation algorithm for the vertex guarding of orthogonal polygons, where $OPT$ denotes the size of an optimal solution.

The complexity of the terrain guarding problem was a long-standing open problem until King and Krohn~\cite{kingKr2011} showed in 2010 that the both continuous and discrete versions of the terrain guarding problem are \textsc{NP}-hard on arbitrary terrains. Ben-Moshe et al.~\cite{benMoshe2005} gave the first $O(1)$-approximation algorithm for the terrain guarding problem, and a 4-approximation algorithm was later given by Elbassioni et al.~\cite{elbassioniAlgorithmica2011}. Katz and Roisman~\cite{katz2008} gave a 2-approximation algorithm for the problem of guarding the vertices of an orthogonal terrain. Moreover, Durocher et al.~\cite{durocherLM2015} gave a linear-time algorithm for guarding the vertices of an orthogonal terrain under a directed visibility model. Gibson et al.~\cite{gibson2009} gave a polynomial-time approximation scheme (PTAS) for the discrete version of the terrain guarding problem, and a PTAS for the continuous version of the problem was given by Friedrichs et al.~\cite{friedrichs2014}. To our knowledge, however, the complexity of the terrain guarding problem on orthogonal terrains remains open. We note that the hardness result of King and Krohn~\cite{kingKr2011} does not seem to be applied to the problem on orthogonal terrains; see~\cite{durocherLM2015} for a discussion on this.

\noindent\paragraph{\bf Our Result.} In this paper, we are interested in a variant of the discrete terrain guarding problem in which the sets $P$ and $G$ are both subsets of the vertices of $T$. For an orthogonal terrain $T$, we denote the set of vertices and edges of $T$ by $V(T)$ and $E(T)$, respectively. A vertex $u\in V(T)$ is called \emph{convex} (resp., \emph{reflex}) if the angle formed by the edges incident to $u$ above $T$ is $\pi/2$ (resp., $3\pi/2$) degrees. We denote the set of convex and reflex vertices of $T$ by $V_C(T)$ and $V_R(T)$, respectively. In this paper, we consider the discrete terrain guarding problem, where $P=V_C(T)$ and $G=V_R(T)$.

\begin{definition}{\bf \emph{(The Vertex Terrain Guarding (VTG) Problem on Orthogonal Terrains).}}
Given an orthogonal terrain $T$, compute a subset $S\subseteq V_R(T)$ of minimum cardinality that guards all the vertices in $V_C(T)$. That is, every vertex $u\in V_C(T)$ is seen by at least one reflex vertex in $S$. A reflex vertex $v$ guards (or, sees) a convex vertex $u$ if and only if every point on the line segment $uv$ (except its endpoints) lies strictly above $T$. 
\end{definition}

In this paper, we give an exact algorithm for the VTG problem on orthogonal terrains. We first give an Integer Program (IP) formulation of the VTG problem and will then prove that the solution to the relaxation of the IP is integral. To this end, we describe a permutation of the rows and columns of the corresponding constraint matrix of the IP relaxation, and prove that the resulting matrix is totally balanced. We remark that in the VTG problem \begin{inparaenum}[(i)] \item the line of visibility lies strictly above the terrain, and \item the objective is to guard only the convex vertices of the terrain (not all the vertices). \end{inparaenum}

\begin{figure}[t]
\centering%
\includegraphics[width=.60\textwidth]{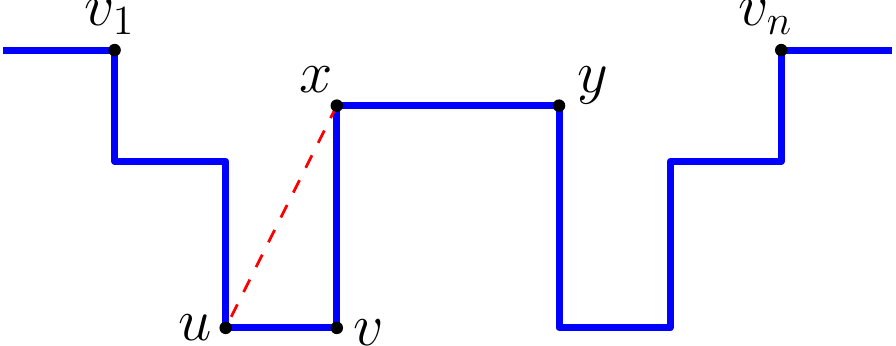}
\caption{An orthogonal terrain $T$. The vertices $u, v, x$ and $y$ are left convex, right convex, left reflex and right reflex, respectively. Moreover, $x$ can see vertex $u$, but it does not see vertex $v$.}
\label{fig:definitions}%
\end{figure}

\noindent\paragraph{\bf Notations.} Throughout this paper, let $T$ be an orthogonal terrain with $\lvert V(T)\rvert=n$. We denote the $x$- and $y$-coordinates of a point $p$ on $T$ by $x(p)$ and $y(p)$, respectively. We let $V(T)=\{v_1,\dots,v_n\}$ to denote the set of vertices of $T$ ordered from left to right, and $E(T)=\{e_1=(v_1,v_2),\dots,e_{n-1}=(v_{n-1},v_n)\}$ to be the set of edges of $T$ induced by the vertex set $V(T)$. Similar to Durocher et al.~\cite{durocherLM2015}, we partition the vertices of $T$ into four equivalences classes left or right endpoints of a horizontal edge of $T$, and whether the vertex is convex or reflex. We use $V_{LC}(T)$, $V_{RC}(T)$, $V_{LR}(T)$ and $V_{RR}(T)$ to respectively denote the left convex, right convex, left reflex, and right reflex subsets of $V(T)$. See Figure~\ref{fig:definitions} for an example of these definitions.

For the rest of this paper, we use terms ``terrain'' and ``guard'' to refer to an orthogonal terrain and a reflex vertex guard, respectively, unless otherwise stated. We assume that the visibility line segment $pq$ lies strictly above $T$, for every two points $p$ and $q$ on $T$ that see each other. Moreover, we assume that the leftmost and rightmost edges of $T$ are two horizontal rays starting from $v_1$ and $v_n$, respectively (see Figure~\ref{fig:definitions} for an example). For a reflex vertex $x$ of $T$, we denote the convex vertex directly below $x$ by $B(x)$.




\section{The VTG Problem: An Exact Algorithm}
\label{sec:algorithm}
In this section, we give an exact algorithm for the VTG problem on an orthogonal terrain $T$ with $n$ vertices. Our algorithm is based on an Integer Program (IP) formulation of the problem. We first give an IP formulation of the problem, and will then show that the solution to the relaxation of the IP is integral. To this end, we show a permutation of the rows and columns of the constraint matrix of the relaxation that results in a totally balanced matrix. We first describe some properties of orthogonal terrains.

\begin{observation}[Durocher et al.~\cite{durocherLM2015}]
\label{obs:visibilityCondition}
Let $u$ be a reflex vertex of $T$. If $u$ is right reflex and sees a right convex vertex $v$ of $T$, then $x(u)<x(v)$ and $y(u)>y(v)$. Similarly, if $u$ is left reflex and sees a left convex vertex $v$ of $T$, then $x(u)>x(v)$ and $y(u)>y(v)$.
\end{observation}

\begin{lemma}[Order Claim~\cite{benMoshe2005}]
\label{lem:orderClaim}
Let $p, q, r$ and $s$ be four vertices of a terrain $T$ such that $x(p)<x(q)<x(r)<x(s)$. If $p$ sees $r$ and $q$ sees $s$, then $p$ sees $s$.
\end{lemma}

Let $u$ be a right reflex vertex and let $v$ be a left convex vertex of $T$. If $v$ is not adjacent to $u$, then it is easy to see that the line segment $uv$ would intersect the area below $T$. Moreover, if $v$ is adjacent to $u$ then clearly $u$ cannot see $v$ because the line segment $uv$ is not strictly above $T$. This means that a right reflex vertex of $T$ cannot see any left convex vertex of $T$. An analogous argument can be made to show that a left reflex vertex of $T$ cannot see any right convex vertex of $T$. Therefore, we have the following lemma.
\begin{lemma}
\label{lem:noRRorLLVisibility}
Let $u$ be a reflex vertex of $T$. If $u$ is right reflex (resp., left reflex), then $u$ cannot see any left convex (resp., right convex) vertex of $T$.
\end{lemma}

\begin{wrapfigure}{r}{0.35\textwidth}
\centering
\vspace{-10pt}
\includegraphics[scale=.60]{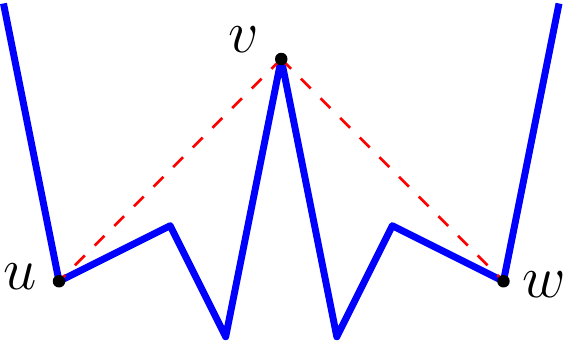}%
\caption{An arbitrary terrain in which the reflex vertex $v$ can see the left convex vertex $u$ and the right convex vertex $w$.}
\label{fig:arbitrary}%
\vspace{-30pt}
\end{wrapfigure}%

\begin{lemma}
\label{lem:dominantVertex}
Let $u$ be a right convex vertex and $v$ be a left convex vertex of a terrain $T$. Then, there is no reflex vertex of $T$ that sees both $u$ and $v$.
\end{lemma}
\begin{proof}
The proof follows from Lemma~\ref{lem:noRRorLLVisibility} and the fact that $\{V_{LR}(T), V_{RR}(T)\}$ is a partition for $V_R(T)$.
\end{proof}

We remark that if $T$ is not orthogonal, then Lemma~\ref{lem:dominantVertex} does not hold. That is, it is possible for an arbitrary polygon that a reflex vertex guards both a left convex and a right convex vertex; see Figure~\ref{fig:arbitrary} for an example.

\noindent\paragraph{\bf Algorithm.} Let $P=V_{LC}(T)\cup V_{RC}(T)$ (i.e., the set of vertices that must be guarded), where $\lvert P\rvert = k$, and let $G=V_R(T)$ (i.e., the set of potential guards), where $\lvert G\rvert =k'$. The VTG problem can be formulated as follows:

\begin{align}%
\label{ip:vtg}%
\text{minimize }               & \sum_{i=1}^{k'}x_{r_i}\\
\nonumber \text{subject to }   & \sum_{r_i\in S(c_j)}x_{r_i}\geq 1, &\forall j=1,2,\dots, k\\%
\nonumber                      & x_{r_i}\in \{0, 1\} &\forall r_i\in V_R(T),%
\end{align}%

\noindent where $x_{r_i}$ is a binary variable for the reflex vertex $r_i\in V_R(T)$ indicating whether $r_i$ is chosen as a guard, and $S(c_j)$ is the set of all reflex vertices that see the convex vertex $c_j$. Consider a binary matrix $M\in \{0, 1\}^{\lvert P\rvert\times \lvert G\rvert}$. Matrix $M$ is called \emph{totally balanced} if \begin{inparaenum}[(i)] \item $M$ does not contain a square submatrix with all row and column sums equal to 2, and \item $M$ has no identical columns. \end{inparaenum} Moreover, $M$ is in \emph{standard greedy form} if it does not contain
\begin{equation}
I=
\begin{bmatrix}
\label{mtx:badOne}
1 & & 1\\
1 & & 0
\end{bmatrix}
\end{equation}
as an induced submatrix. Finally, Hoffman et al.~\cite{hoffman1985} proved the following result.

\begin{theorem}[Hoffman et al.~\cite{hoffman1985}]
\label{thm:balancedAndGreedy}
Matrix $M$ is totally balanced if and only if $M$ can be converted into greedy standard form by permuting its rows and columns.
\end{theorem}

Let $M\in \{0, 1\}^{\lvert P\rvert\times \lvert G\rvert}$ be the incident matrix corresponding to the IP~\eqref{ip:vtg} of an instance of the VTG problem. Consider the permutation of rows and columns on $M$ such that:
\begin{itemize}
\item the first $\lvert V_{RC}(T)\rvert$ rows of $M$ from top to bottom correspond to the convex vertices in $V_{RC}(T)$ ordered from left to right,
\item the rest of the rows of $M$ from top to bottom (i.e., following the first $\lvert V_{RC}(T)\rvert$ rows) correspond to the convex vertices in $V_{LC}(T)$ ordered from right to left,
\item the first $\lvert V_{RR}(T)\rvert$ columns of $M$ from left to right correspond to the reflex vertices in $V_{RR}(T)$ ordered from right to left, and
\item the rest of the columns of $M$ from left to right (i.e., following the first $\lvert V_{RR}(T)\rvert$ columns) correspond to the reflex vertices in $V_{LR}(T)$ ordered from left to right.
\end{itemize}

Figure~\ref{fig:matrixDecompose} shows an illustration of permuting the rows and columns of $M$; note that the arrows indicate the direction by which the corresponding vertices are visited on $T$. We prove that the resulting matrix $M$ is totally balanced.

\begin{figure}[t]
\centering%
\includegraphics[width=.55\textwidth]{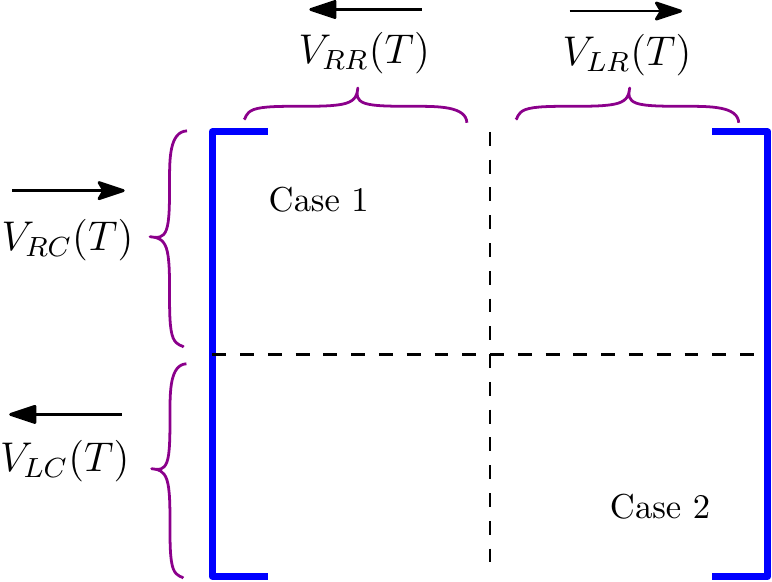}
\caption{An illustration of permuting the rows and columns of matrix $M$. An arrow indicates the direction by which the corresponding vertices are visited on $T$.}
\label{fig:matrixDecompose}%
\end{figure}

\begin{lemma}
\label{lem:totallyBalanced}
Matrix $M$ is totally balanced.
\end{lemma}
\begin{proof}
By Theorem~\ref{thm:balancedAndGreedy}, it suffices to show that $M$ is in standard greedy form. Suppose for a contradiction that $M$ is not in standard greedy form; that is, $M$ contains the matrix $I$ in~\eqref{mtx:badOne} as an induced submatrix. Let $i_1, i_2$ (resp., $j_1, j_2$) be the indices of the rows (resp., columns) of $I$ in $M$. Moreover, let $p_1$ and $p_2$ (resp., $g_1$ and $g_2$) be the convex vertices (resp., reflex vertices) of $T$ corresponding to $i_1$ and $i_2$ (resp., $j_1$ and $j_2$), respectively. First, we note that either $i_1, i_2 \leq \lvert V_{RC}(T)\rvert$ or $i_1, i_2> \lvert V_{RC}(T)\rvert$ because the reflex vertex $g_1$ cannot see a right convex and a left convex of $T$ at the same time by Lemma~\ref{lem:dominantVertex}.

\noindent{\bf Case 1: $i_1, i_2 \leq \lvert V_{RC}(T)\rvert$.} If $i_1, i_2 \leq \lvert V_{RC}(T)\rvert$, then $p_1, p_2\in V_{RC}(T)$. By Lemma~\ref{lem:noRRorLLVisibility}, no left reflex vertex of $T$ can see $p_1$. This means that both $g_1$ and $g_2$ are right reflex vertices of $T$; that is, $j_1, j_2 \leq \lvert V_{RR}(T)\rvert$. But, if both $g_1$ and $g_2$ are right reflex vertices of $T$, then we get the ordering $x(g_2)<x(g_1)<x(p_1)<x(p_2)$ by Observation~\ref{obs:visibilityCondition}; see Figure~\ref{fig:mainLemma}(a) for an illustration. By Lemma~\ref{lem:orderClaim}, vertex $g_2$ must see vertex $p_2$, which is a contradiction.

\begin{figure}[t]
\centering%
\includegraphics[width=.90\textwidth]{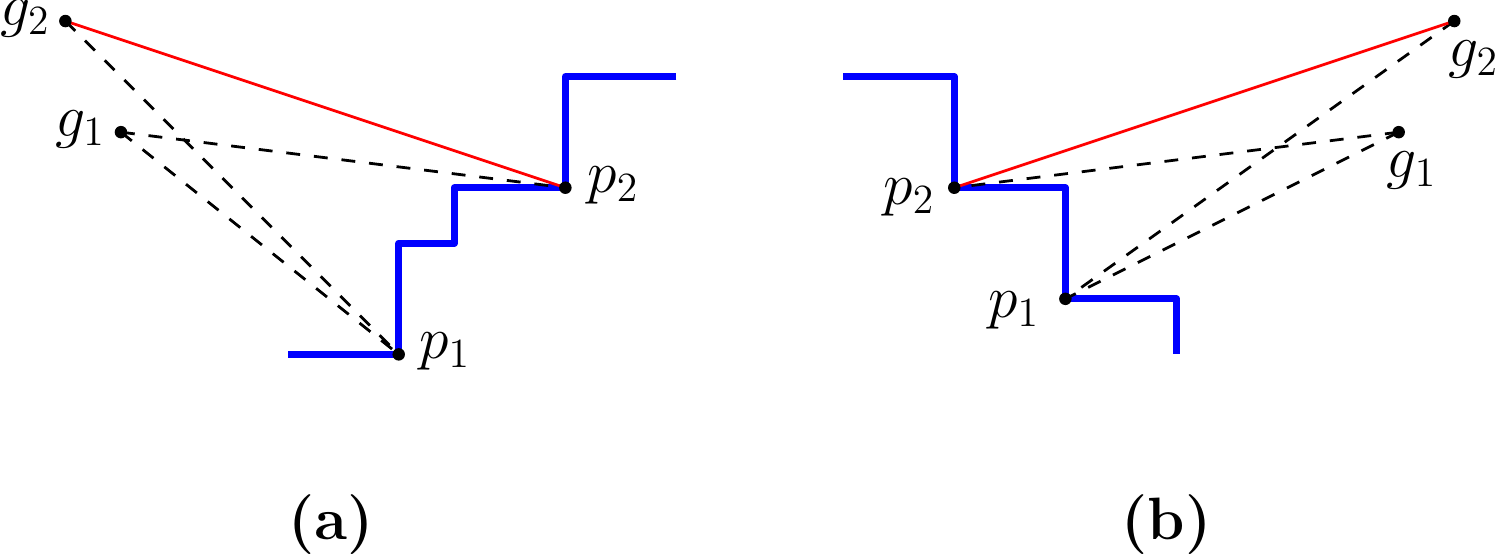}
\caption{An illustration in support of the proof of Lemma~\ref{lem:totallyBalanced}.}
\label{fig:mainLemma}%
\end{figure}

\noindent{\bf Case 2: $i_1, i_2 > \lvert V_{RC}(T)\rvert$.} If $i_1, i_2 > \lvert V_{RC}(T)\rvert$, then $p_1, p_2\in V_{LC}(T)$. By Lemma~\ref{lem:noRRorLLVisibility}, no right reflex vertex of $T$ can see $p_1$. This means that both $g_1$ and $g_2$ are left reflex vertices of $T$; that is, $j_1, j_2 > \lvert V_{RR}(T)\rvert$. But, if both $g_1$ and $g_2$ are left reflex vertices of $T$, then we get the ordering $x(p_2)<x(p_1)<x(g_1)<x(g_2)$ by Observation~\ref{obs:visibilityCondition}; see Figure~\ref{fig:mainLemma}(b) for an illustration. By Lemma~\ref{lem:orderClaim}, vertex $g_2$ must see vertex $p_2$, which is a contradiction.

By Cases 1 and 2, we conclude that $M$ is in standard greedy form and, therefore, it is a totally balanced matrix. This completes the proof of the lemma.
\end{proof}

We know that for a totally balanced matrix $M$, the polyhedron $\{x \geq 0 : Mx \geq 1\}$ is integral. Moreover, Kolen~\cite{kolen1982} gave an efficient algorithm for an optimal integer solution for the IP of form~\eqref{ip:vtg}. Therefore, we have the main result of this paper:

\begin{theorem}
\label{thm:mainResult}
There exists a polynomial-time exact algorithm for the VTG problem on orthogonal terrains.
\end{theorem}


\section{Conclusion}
\label{sec:conclusion}
In this paper, we gave an exact algorithm for the problem of guarding the convex vertices of an orthogonal terrain $T$ with the minimum number of reflex vertices (i.e., the VTG problem). The complexity of the problem remains open if the objective is to guard \emph{all} the vertices of $T$ with the minimum number of (not necessarily reflex) vertex guards.

\bibliographystyle{plain}
\bibliography{terrainRef}

\end{document}